\documentclass[11pt]{article}
\usepackage{graphicx} 
\usepackage[english]{babel}
\usepackage{csquotes}
\usepackage[top=2cm,bottom=3.5cm,left=2.5cm,right=2.5cm]{geometry}
\usepackage{amsmath, amsfonts, amsthm, MnSymbol, algorithm, algpseudocode, graphicx, enumerate, mathtools}
\usepackage{xcolor}
\usepackage{hyperref}
\hypersetup{colorlinks,allcolors=blue,breaklinks=true}
\usepackage{breakurl}
\usepackage[backend=bibtex, style=alphabetic, maxbibnames=99]{biblatex}

\theoremstyle{definition}
\newtheorem{definition}{Definition}[section]
\newtheorem{remark}[definition]{Remark}
\newtheorem{theorem}{Theorem}[section]
\newtheorem{corollary}[theorem]{Corollary}
\newtheorem{lemma}[theorem]{Lemma}
\newtheorem{fact}[theorem]{Fact}

\newcommand{\N}{\mathbb{N}}

\newcommand{\Z}{\mathbb{Z}}
\newcommand{\R}{\mathbb{R}}
\newcommand{\Prob}{\mathbb{P}}

\newcommand{\bigo}{\mathcal{O}}

\newcommand{\innerproduct}[2]{\langle #1, #2 \rangle}

\newcommand{\abs}[1]{\lvert #1 \rvert}
\newcommand{\Bigabs}[1]{\Big\lvert #1 \Big\rvert}

\newcommand{\norm}[1]{\lVert #1 \rVert}

\DeclarePairedDelimiter\bra{\langle}{\rvert}
\DeclarePairedDelimiter\ket{\lvert}{\rangle}

\DeclarePairedDelimiterX\braket[2]{\langle}{\rangle}{#1\,\delimsize\vert\,\mathopen{}#2}

\title{Lattice Based Crypto Breaks in a Superposition of Spacetimes}
\author{
Divesh Aggarwal\\National University of Singapore\\ \texttt{divesh@comp.nus.edu.sg} \and 
Shashwat Agrawal\\ Indian Institute of Technology Delhi  \\ \texttt{csz248012@iitd.ac.in} \and
Rajendra Kumar\\ Indian Institute of Technology Delhi  \\ \texttt{rajendra@cse.iitd.ac.in}}

\date{}

\begin{document}

\maketitle

\begin{abstract}
We explore the computational implications of a superposition of spacetimes, a phenomenon hypothesized in quantum gravity theories. This was initiated by Shmueli (2024) where the author introduced the complexity class $\mathbf{BQP^{OI}}$ consisting of promise problems decidable by quantum polynomial time algorithms with access to an oracle for computing order interference. In this work, it was shown that the Graph Isomorphism problem and the Gap Closest Vector Problem (with approximation factor $\bigo(n^{3/2})$) are in $\mathbf{BQP^{OI}}$. We extend this result by showing that the entire complexity class $\mathbf{SZK}$ (Statistical Zero Knowledge) is contained within $\mathbf{BQP^{OI}}$. This immediately implies that the security of numerous lattice based cryptography schemes will be compromised in a computational model based on superposition of spacetimes, since these often rely on the hardness of the Learning with Errors problem, which is in $\mathbf{SZK}$.
\end{abstract}

\section{Introduction}

The advent of Quantum Computing has compelled computer scientists and information theorists across all domains to question the hardness of problems given access to this tool. Cryptography is one of the most prominent areas affected by Quantum Computing. Indeed, the current popular public key cryptosystems that are mostly based on the hardness of integer factoring (RSA) or the discrete logarithm (Diffie Hellman) can be rendered obsolete (atleast in theory) due to Shor's quantum algorithms \cite{shor} which can efficiently solve both these problems. This has led to the birth of post-quantum cryptography, i.e. building cryptographic primitives based on the hardness of problems which are believed to be hard to solve efficiently even by quantum computers. Currently the most promising candidates for this purpose are lattice problems, such as the Gap Shortest Vector Problem, Lattice Isomorphism, Learning with Errors, etc.
\cite{regevLatticesLearningErrors2009,ABD+CRYSTALSKyberVersion022021, NIST2022,cryptoeprint:2021/1332} But what happens if we allow more computational power to our model motivated by some phenomena in physics? 

\textbf{Superposition of Spacetimes.}
Quantum Computers and Quantum Information Theory have been designed on the principles of Quantum Mechanics (QM). Quantum Mechanics (together with Quantum Field Theory) describes the interaction between the fundamental forces: strong nuclear force, weak nuclear force and electromagnetism. Apart from this, the theory of General Relativity (GR) deals with the fourth fundamental force: gravity. A unified theory that can generalize both QM and GR and can be experimentally verified is one of the biggest open question in physics. Such a theory is referred to as Quantum Gravity. Although there is no experimental evidence yet to support such a theory, there are some physical phenomena expected to take place within any valid quantum gravity theory. One of the fundamental expected phenomena unique to quantum gravity is the possibility for a superposition of spacetime geometries \cite{qg_effects_1,qg_effects_2,qg_effects_3,qg_effects_4,qg_effects_5,qg_effects_6}. This is hypothesized by combining two phenomena: quantum superposition of all dynamic quantities (due to QM) and variation in the geometry of spacetime (due to GR). 

Consider the following thought experiment by Zych, Costa, Pikovski and Brukner \cite{qg_experiment} (with small variations)\footnote{this is borrowed from Section 1.1 in \cite{bqp_oi}}: Let $S_M$ and $S_U$ be two spatially isolated physical systems. $S_M$ contains an object of mass $M$ and two possible locations $x_0, x_1$ for the object. $S_U$ consists of a quantum register $R$ in a state $\ket{\psi}$, two atomic clocks $C_0, C_1$ and a central processing unit $C_c$ that gets signals from the clocks, such that it executes the unitary $U_i$ on $R$ when getting a signal from clock $C_i$. Further, when $M$ is at location $x_0$, due to gravitational time dilation, clock $C_1$ ticks slower (than $C_0$), and when $M$ is at $x_1$, clock $C_0$ ticks slower. As a consequence, when $M$ is at $x_0$, the register $R$ gets the state $U_1U_0\ket{\psi}$, and when $M$ is at $x_1$, it has the state $U_0U_1\ket{\psi}$. However, what can we say about the joint state of the system when $M$ is in superposition of $x_0$ and $x_1$?

The gravitational decoherence hypothesis \cite{gravitational_decoherence_1, gravitational_decoherence_2, gravitational_decoherence_3, gravitational_decoherence_4} says that a quantum superposition of a massive object will fundamentally decohere, turning into a probabilistic mixture of classical states. In this case the joint system ($S_M, S_U$) will be in the uniform mixed state:
$$\{\ket{x_0}_{S_M}\otimes U_1U_0\ket{\psi}_{S_U},\ket{x_1}_{S_M}\otimes U_0U_1\ket{\psi}_{S_U}\}$$
On the other hand, gravitationally-induced entanglement (GIE) \cite{gie_1,gie_2} says that not only a massive object can be in a superposition, but that the gravitational field can be used to mediate entanglement. In this scenario the joint system ($S_M, S_U$) will be in the pure quantum state:
$$\frac{1}{\sqrt{2}}\ket{x_0}_{S_M}\otimes U_1U_0\ket{\psi}_{S_U} + \frac{1}{\sqrt{2}}\ket{x_1}_{S_M}\otimes U_0U_1\ket{\psi}_{S_U}$$

It is not hard to see that both of these possibilities can be efficiently simulated using standard quantum computing. However there is a third possibility; correlation without entanglement (this is known to be possible, e.g., quantum discord \cite{quantum_discord}), i.e. a massive object can be maintained in a coherent superposition, as opposed to gravitational decoherence, but gravity alone cannot create entanglement, as opposed to GIE. In this scenario, the joint system ($S_M, S_U$) will be in the following state (upto normalization):
$$(\ket{x_0} + \ket{x_1})_{S_M}\otimes (U_1U_0\ket{\psi} + U_0U_1\ket{\psi})_{S_U}$$

\textbf{Computational Order Interference.}
Our focus is on the computational implications in the last scenario. Notice that the norm of $(U_1U_0\ket{\psi} + U_0U_1\ket{\psi})$ depends on $U_1, U_0, \ket{\psi}$. In fact, this vector might be the zero vector, for eg. if $U_0 = X, U_1 = Z$ then
$$(U_1U_0\ket{\psi} + U_0U_1\ket{\psi}) = (XZ + ZX)\ket{\psi} = \vec{0}$$

More generally, \cite{costa} shows that an unconditional generation of such a state is not possible as it cannot yield a Process Matrix. The Process Matrix Formalism \cite{pmf} (PMF) is a mathematical framework for defining quantum processes that are not constrained by a causal execution order. A transformation which is not a process matrix is in particular not a valid quantum process in standard quantum mechanics.

However \cite{bqp_oi} considers a relaxation of the (unconditional) generation task; in particular given $(U_0, U_1, \ket{\psi})$, they allow computing the above state in time inversely proportional to the norm of the state. \footnote{In addition, the computation time also depends on the amount of phase alignment between execution orders; but that is not very important for the present discussion.} This state is called the Order Interference state corresponding to $(U_0, U_1, \ket{\psi})$. More generally, if we are given $m$ unitaries, the order interference state is defined as

\begin{align*}
    \text{OI}(\{U_i\}_{i\in[m]},\ket{\psi}) := \sum_{\sigma\in S_m}\Big(\prod_{i\in[m]}U_{\sigma^{-1}(i)}\Big)\ket{\psi}
\end{align*}
where $\prod_{i\in[m]}U_{\sigma^{-1}(i)}$ denotes $U_{\sigma^{-1}(m)}\ldots U_{\sigma^{-1}(1)}$

\textbf{Prior Work.}
$\mathbf{BQP}$ is the computational complexity class associated with promise problems solvable using a uniform family of polynomially sized quantum circuits. It represents the class of problems solvable efficiently using quantum computers. Shmueli \cite{bqp_oi} considered a computational model wherein such quantum circuits are given an additional power; the ability to use an oracle which computes order interference of a given set of unitaries on a quantum state. The complexity class corresponding to such computational power is called $\mathbf{BQP^{OI}}$. In the same work, Shmueli also introduced the Sequentially Invertible Statistical Difference (SISD) problem, which is a variant of the well-known Statistical Difference problem. SISD for certain parameters was shown to belong in $\mathbf{BQP^{OI}}$. Note that like Statistical Difference, SISD is also a completely classical problem. Further, he considered two important problems in theoretical computer science, namely the Graph Isomorphism problem and the Gap Closest Vector problem (the latter is a crucial problem for post-quantum cryptography), which are not known to be decidable by efficient quantum algorithms. Both of these were shown to have a classical reduction to SISD (for suitable parameters), which implies that they are decidable in $\mathbf{BQP^{OI}}$. 

\textbf{Our Result.}
Many cryptographic-relevant problems as well as other hard algorithmic problems have been shown to be in the complexity class $\mathbf{SZK}$, which constitutes promise problems that have interactive proof systems where a prover can convince a verifier of a ``yes'' instance without revealing any additional information beyond the validity of the statement. This zero-knowledge property is statistical, meaning that the verifier's view can be simulated to within a negligible statistical distance. The hardest problems in $\mathbf{SZK}$ are currently not known to be decidable by efficient quantum algorithms. Sahai and Vadhan \cite{sv} introduced the Statistical Difference (SD) problem and showed that it is a complete problem for $\mathbf{SZK}$.

The problems considered by Shmueli; namely the Graph Isomorphism problem and the Gap Closest Vector problem (for suitable approximation factor) have been shown to be in $\mathbf{SZK}$. In this work, we extend the prior work by showing that the entire class $\mathbf{SZK}$ is contained in $\mathbf{BQP^{OI}}$. 

\begin{theorem}(Informal Main Theorem: $\mathbf{SZK} \subseteq \mathbf{BQP^{OI}}$) For every promise problem $\Pi \in \mathbf{SZK}$, there exists a polynomial time quantum algorithm with access to a computable order interference oracle $\mathcal{OI}$ which decides $\Pi$ correctly with high probability.
\end{theorem}

We achieve this by giving a poly-time Karp reduction from SD itself to SISD. The idea for this reduction is very straightforward, and the intuition is as follows: Given circuits $C_0, C_1$ as inputs for SD, we construct circuit families $\mathfrak{C}^0 = \{C_{i}^0\}_{i\in[l]}, \mathfrak{C}^1 = \{C_{i}^0\}_{i\in[l]}$ which are ``invertible per hard-wired randomness'', i.e. valid instances of SISD\footnote{the SISD problem is defined in Definition~\ref{SISD}}. The output of $\mathfrak{C}^b$ is $(x',C_b(x))$ for $x',x$ being independent and uniformly random strings from the input space of $C_b$. As a consequence, the marginal distribution of the second element in the above tuple is exactly $D(C_b)$ and thus the statistical difference between the distributions of $C_0,C_1$ is the same as the statistical difference between the distributions $\mathfrak{C}^0, \mathfrak{C}^1$.

A major stakeholder in post-quantum cryptography is lattice-based cryptography, and the security of most of the lattice-based cryptographic primitives is based on the hardness of the Learning with Errors (LWE) problem (see, e.g. an overview in \cite{crypto_primitives_on_LWE}). We elaborate on the fact that the LWE problem is in $\mathbf{SZK}$ for parameters where its hardness is based on the hardness of lattice problems (and those parameters can be used in previously known public-key encryption schemes). With this, $\mathbf{SZK}$ having efficient algorithms in a computational model can be viewed in two ways. On one hand, if one believes that $\mathbf{SZK}$ is hard, then perhaps the complexity class $\mathbf{BQP^{OI}}$ is intractable as a computational model. On the other hand, if superposition of spacetimes not only exists in the universe, but we are also able to harness this phenomenon in realizing computational order interference, then we can efficiently solve $\mathbf{SZK}$ problems and thereby attack lattice-based cryptographic primitives. 

\section{Preliminaries}

\textbf{Notations.} $\N$ denotes the set of natural numbers. For $n\in\N$, $[n]$ denotes the set $\{1,2,\ldots,n\}$. For $m\in\N$, $S_m$ denotes the set of all permutations on $[m]$. $\{0,1\}^* = \cup_{i\in\N}\{0,1\}^i$ is the set of all binary strings of finite length. For a binary string $x$ of length $>j$, $x_{1\ldots j}$ denotes the $j$ length prefix of $x$. $\oplus$ denotes the bit-wise xor of two strings of equal length. $\innerproduct{}{}$ denotes the dot product between two $n$-tuples. $\Prob$ denotes the probability function. For a set $S$, we denote $x\xleftarrow{\$}S$ for a sample $x$ chosen uniformly at random from $S$. $U(S)$ denotes the uniform distribution on set $S$. For a circuit $C:\{0,1\}^{k'}\to\{0,1\}^k$, we denote its distribution $D(C) := \{C(x)\ |\ x\xleftarrow{\$}\{0,1\}^{k'}\}$.

\subsection{Order Interference}

We first recall the complexity class $\mathbf{BQP}$ which is considered the quantum analogue to the classical complexity class $\mathbf{BPP}$ (Bounded-error Probabilistic Polynomial time). It represents the class of problems that can be efficiently solved by quantum computers.

\begin{definition}(The Complexity Class $\mathbf{BQP}$)
    The complexity class $\mathbf{BQP}$ (short for Bounded Error Quantum Polynomial Time) consists of all promise problems $\Pi = (\Pi_{\text{YES}},\Pi_{\text{NO}})$ for which there exists a polynomial time uniform family of quantum circuits $\{Q_n\}_{n\in\N}$ such that
    \begin{itemize}
        \item For all $n\in\N$, $Q_n$ takes $n$ qubits as input and outputs one bit.
        \item If $x\in\Pi_{\text{YES}}$, then $\Prob[Q_{\abs{x}}(x)=1] \geq 2/3$
        \item If $x\in\Pi_{\text{NO}}$, then $\Prob[Q_{\abs{x}}(x)=1] \leq 1/3$
    \end{itemize}
\end{definition}

\cite{bqp_oi} defined a specific computational task motivated by superposition of spacetimes, a possible phenomenon in Quantum Gravity. 

\begin{definition}
    (Order Interference) Let $U_1,\ldots,U_m$ be $n$-qubit unitaries and $\ket{\psi}$ be a $n$-qubit quantum state. We define the order interference corresponding to $(\{U_i\}_{i\in[m]},\ket{\psi})$ as the $2^n$ dimensional vector corresponding to the superposition of all possible sequences of unitaries $\{U_i\}_{i\in[m]}$ acting on $\ket{\psi}$
    \begin{align*}
        \text{OI}(\{U_i\}_{i\in[m]},\ket{\psi}) := \sum_{\sigma\in S_m}\Big(\prod_{i\in[m]}U_{\sigma^{-1}(i)}\Big)\ket{\psi}
    \end{align*}
    where $\prod_{i\in[m]}U_{\sigma^{-1}(i)}$ denotes $U_{\sigma^{-1}(m)}\ldots U_{\sigma^{-1}(1)}$
\end{definition}

\begin{remark}
    Shmueli actually calls the above a \emph{uniform} order interference state, and further gives a definition for a generalized order interference state. However, we only need to use the above definition in this work.
\end{remark}

Note that $\norm{\text{OI}(\{U_i\}_{i\in[m]},\ket{\psi})} \in [0,m!]$. Thus $\text{OI}(\{U_i\}_{i\in[m]},\ket{\psi})$ may not be a quantum state in general, even after normalization. Indeed, for $U_1 = X, U_2 = Z$, ($X,Z$ being Pauli operators) we have 
$$\text{OI}(\{U_1,U_2\},\ket{\psi}) = (XZ + ZX)\ket{\psi} = \vec{0}$$
so this cannot be a valid quantum state. 

As mentioned in the introduction, \cite{bqp_oi} considers a relaxation of the (unconditional) generation task of OI; in particular they allow the order interference task a computation time dependent on the norm of the OI state and the amount of phase alignment between execution orders (or spacetimes). For every classical state $\ket{x}, x\in\{0,1\}^n$, the amount of phase alignment is how similar are the phases of a particular $\ket{x}$ between the results of differing execution orders $\Big(\prod_{i\in[m]}U_{\sigma^{-1}(i)}\Big)\ket{\psi}$, over all possible $x\in\{0,1\}^n$.

More precisely they provide oracle access to the OI task in the following manner:

\begin{definition}
    (Definition 4.2 in \cite{bqp_oi}) Let $\{U_i\}_{i\in[m]}$ be a set of unitaries acting on $n$-qubit states, and let $\ket{\psi}$ be a $n$-qubit quantum state. An oracle $\mathcal{OI}$ takes as input a triplet: $(\{U_i\}_{i\in[m]},\ket{\psi},\lambda)$ for $\lambda\in\N$ and with probability
    $$\Bigg(\frac{\sum_{x\in\{0,1\}^n{\abs{\sum_{\sigma\in S_m}{\alpha_{x,\sigma}}}}}}{\sum_{x\in\{0,1\}^n{\sum_{\sigma\in S_m}{\abs{\alpha_{x,\sigma}}}}}}\Bigg) \Bigg(\frac{\frac{\norm{\text{OI}(\{U_i\}_{i\in[m]},\ket{\psi})}}{m!}}{\frac{\norm{\text{OI}(\{U_i\}_{i\in[m]},\ket{\psi})}}{m!}+\frac{1}{\lambda}}\Bigg)$$
    obtains the normalized state of $\text{OI}(\{U_i\}_{i\in[m]},\ket{\psi})$. Moreover it takes a time complexity of $\bigo(\lambda+\sum_{i\in[m]}{\abs{U_i}})$
\end{definition}

\begin{remark}
    Shmueli also defined two variants of the above oracle which are less restrictive; one where the probability of outputting a valid OI state is not dependent on the norm of the OI state, and another where the same is not dependent on the phase alignment. 

    For the latter definition of $\mathcal{OI}$, he remarks that the complexity classes $\mathbf{NP}$ and $\mathbf{QCMA}$ are contained in $\mathbf{BQP^{OI}}$.\footnote{This has not been cited in \cite{bqp_oi}, but personally communicated to Shmueli by Scott Aaronson and Joseph Carolan.}
    
    Again, for our purpose, the above definition suffices.
\end{remark}

Next we define the notion of choice interference. Intuitively, given $m$ unitaries and a quantum state, this is a superposition of all choices of unitaries to act on the quantum state.

\begin{definition}(Definition 5.1 in \cite{bqp_oi})
    Let $\{U_i\}_{i\in[m]}$ be a set of unitaries acting on $n$-qubit states, and let $\ket{\psi}$ be a $n$-qubit quantum state. The choice interference corresponding to $(\{U_i\}_{i\in[m]},\ket{\psi})$ is the $2^n$ dimensional vector which is the superposition of all unitaries $\{U_i\}_{i\in[m]}$ acting on $\ket{\psi}$
    \begin{align*}
        \text{CI}(\{U_i\}_{i\in[m]},\ket{\psi}) := \sum_{i\in[m]}U_i\ket{\psi}
    \end{align*}
\end{definition}

Shmueli then showed that a normalization of the choice interference state can be obtained using an $\mathcal{OI}$ oracle. 

\begin{lemma}\label{OI implies CI}(Lemma 5.1 in \cite{bqp_oi})
    Let $\{U_i\}_{i\in[m]}$ be a set of unitaries acting on $n$-qubit states, and let $\ket{\psi}$ be a $n$-qubit quantum state. There exists a quantum algorithm $Q_{CI}$ with oracle access to $\mathcal{OI}$ that takes as input a triplet $(\{U_i\}_{i\in[m]},\ket{\psi},\lambda)$ for $\lambda\in\N$ and with probability
    $$\Bigg(\frac{\sum_{x\in\{0,1\}^n{\abs{\sum_{i\in[m]}{\alpha_{x,i}}}}}}{\sum_{x\in\{0,1\}^n{\sum_{i\in[m]}{\abs{\alpha_{x,i}}}}}}\Bigg) \Bigg(\frac{\frac{\norm{\text{CI}(\{U_i\}_{i\in[m]},\ket{\psi})}}{m}}{\frac{\norm{\text{CI}(\{U_i\}_{i\in[m]},\ket{\psi})}}{m}+\frac{1}{\lambda}}\Bigg)$$
    obtains the normalized state of $\text{CI}(\{U_i\}_{i\in[m]},\ket{\psi})$. Here $\alpha_{x,i}$ are the coefficients in $U_i\ket{\psi} := \sum_{x\in\{0,1\}^n}\alpha_{x,i}\ket{x}$. Moreover $Q_{CI}$ takes a time complexity of $\bigo(\lambda+(\log{m})\cdotp\sum_{i\in[m]}{\abs{U_i}})$.
\end{lemma}

\subsection{SZK and Statistical Difference}

\begin{definition}(Statistical Difference) 
    Let $D_0, D_1$ be discrete probability distributions over a set $\Omega$. The statistical difference or the total variational distance between $D_0, D_1$ is defined as the following:
    $$\norm{D_0 - D_1} := \frac{1}{2}\sum_{x\in\Omega}\abs{D_0(x)-D_1(x)}$$
\end{definition}

The statistical difference between two distributions is a measure of how `close' or `far apart' they are. It is easy to see that $\norm{D_0 - D_1} = 0$ iff $D_0, D_1$ are identical distributions, and $\norm{D_0 - D_1} = 1$ iff $D_0, D_1$ have disjoint supports.

\begin{definition}(Fidelity)
    Let $D_0, D_1$ be discrete probability distributions over a set $\Omega$. The fidelity between $D_0, D_1$ is defined as the following:
    $$F(D_0,D_1) := \sum_{x\in\Omega}\sqrt{D_0(x)D_1(x)}$$
\end{definition}

\begin{fact}\label{fidelity and statistical distance}
    Let $D_0, D_1$ be discrete probability distributions over a set $\Omega$. Then
    $$1 - F(D_0,D_1) \leq \norm{D_0 - D_1} \leq \sqrt{1 - F(D_0,D_1)^2}$$
\end{fact}

\textbf{Informal Definition of SZK.} We now give an informal definition for the complexity class $\mathbf{SZK}$ (for the formal definition, refer to say, \cite{szk}). 

The complexity class $\mathbf{SZK}$ consists of promise problems that have an interactive proof system, consisting of a probabilitstic polynomial time algorithm Verifier and a computationally unbounded Prover with the following properties:
\begin{itemize}
    \item Completeness: For ``yes'' instances, the Prover can convince the Verifier to output 1 with high probability.
    \item Soundness: For ``no'' instances, the prover can convince the Verifier to output 1 only with low probability.
    \item Zero-knowledge: The Verifier learns nothing beyond the truth of the statement being proved. This notion of ``learns nothing'' is formally captured by the following: The information seen by the Verifier throughout the interactive protocol can be simulated by a probabilistic polynomial time simulator such that the random variables associated with the simulator's output and the Verifier's view have negligible statistical difference.
\end{itemize}

We now define the Statistical Difference problem introduced in \cite{sv}.
\begin{definition}
    Let $C_0 : \{0,1\}^{k_0} \to \{0,1\}^k, C_0 : \{0,1\}^{k_1} \to \{0,1\}^k$ be a pair of classical circuits with their output space consisting of $k$ bits. Consider the distributions $D(C_b) := \{C_b(x)\ |\ x\xleftarrow{\$}\{0,1\}^{k_b}\}$. For functions $a,b:\N\to[0,1]$, $a(n)\leq b(n)\ \forall n\in\N$, we define the $(a(n),b(n))$-gap Statistical Difference problem $\text{SD}_{a(n),b(n)}$ as a promise problem which takes the circuits $(C_0,C_1)$ as inputs satisfying the following:
    \begin{itemize}
        \item $\Pi_{YES}$ is the set of circuits $(C_0,C_1)$ such that $\norm{D(C_0)-D(C_1)}\leq a(n)$.
        \item $\Pi_{NO}$ is the set of circuits $(C_0,C_1)$ such that $\norm{D(C_0)-D(C_1)}> b(n)$.
    \end{itemize}
    Here $n = \abs{(C_0,C_1)}$ is the size of description of the circuits.
\end{definition}

Sahai and Vadhan demonstrated the importance of this problem by showing that this (for suitable choice of functions $a(n),b(n)$) is a complete problem for the class $\mathbf{SZK}$, i.e. if the statistical difference problem can be solved in polynomial time, all problems having a statistical zero knowledge interactive protocol can be solved in polynomial time.

\begin{theorem}\label{SD_is_SZK_complete} (Theorem 3.1 in \cite{sv})
    $\text{SD}_{\frac{1}{3},\frac{2}{3}}$ is $\mathbf{SZK}$-complete.
\end{theorem}

\subsection{Sequentially Invertible Statistical Difference}

Given a circuit $C:\{0,1\}^{k_{in}}\to\{0,1\}^{k_{out}}$, consider its `output distribution state' (here $N$ is a normalizing factor):
\begin{align*}
    \ket{C} := \frac{1}{\sqrt{N}}\sum_{x\in\{0,1\}^{k_{in}}}{\ket{C(x)}}
\end{align*}

Suppose we could generate the output distribution states $\ket{C_0},\ket{C_1}$ corresponding to the circuits $C_0, C_1$. Then it was observed in \cite{ref:solving_sd_by_quantum_states} that $\text{SD}_{a(n),b(n)}$ can be decided using an efficient quantum algorithm for a `reasonable' gap between $a,b$, specifically all $a(n), b(n)$ satisfying $b(n)^2 - 2a(n) + a(n)^2 \geq \frac{1}{\texttt{poly}(n)}$ for some polynomial $\texttt{poly}(n)$. The quantum algorithm would simply apply the Swap Test (\cite{swap_test_1,swap_test_2}) on $\ket{C_0}, \ket{C_1}$, which estimates the inner product $\bra{C_0}\ket{C_1}$. By the relations between fidelity and statistical difference as in Fact~\ref{fidelity and statistical distance}, one can check that the error in the approximation of the Swap Test can be overcome for deciding $\text{SD}_{a(n),b(n)}$ for the above mentioned gap constraints. In particular, $\text{SD}_{\frac{1}{3},\frac{2}{3}}$ would be solvable in this manner. 

However, there is no known way to efficiently produce such output distribution states for general circuits. Shmueli introduced a variant of the statistical difference problem where the inputs are families of ``sequentially invertible'' circuits. 

\begin{definition}
((r,t,\(\ell\))-invertible circuit sequence, Definition 4.8 in \cite{bqp_oi}).  
For \(r, t, \ell, k \in \mathbb{N}\), an \((r,t,\ell)\)-invertible circuit sequence on \(k\) bits is a sequence of \(\ell\) pairs of randomized circuits $\mathfrak{C} = \{(C_{i,\rightarrow}, C_{i,\leftarrow})\}_{i \in [\ell]}$ such that for every \(i \in [\ell]\):  

\begin{itemize}
    \item For every \(Y \in \{\rightarrow, \leftarrow\}\), the circuit \(C_{i,Y}\) is of size at most \(t\), uses \(r_i \leq r\) random bits, and maps from \(k\) to \(k\) bits, that is:
    \[
    C_{i,Y} : \{0,1\}^k \times \{0,1\}^{r_i} \to \{0,1\}^k.
    \]
    \item The circuit directions are inverses of each other, per hard-wired randomness, that is:
    \[
    \forall z \in \{0,1\}^{r_i}, x \in \{0,1\}^k : x = C_{i,\leftarrow}(C_{i,\rightarrow}(x; z); z).
    \]
    Note that this also implies that for each hard-wired randomness, the circuits act as permutations on the set \(\{0,1\}^k\).
\end{itemize}

For \(r \in \mathbb{N}\), an \(r\)-invertible circuit sequence is an \((r,t,\ell)\)-invertible circuit sequence, for finite and unbounded \(t\) and \(\ell\).
\end{definition}

The notions of sequential invertibility can be combined with the Statistical Difference problem, to define the Sequentially Invertible Statistical Difference (SISD) Problem.

\begin{definition}\label{SISD}
(The Sequentially Invertible Statistical Difference Problem, Definition 4.12 in \cite{bqp_oi}).  
Let \(a, b : \mathbb{N} \to [0,1]\) and \(r : \mathbb{N} \to \mathbb{N} \cup \{0\}\) be functions, such that for every \(n \in \mathbb{N}\), \(a(n) \leq b(n)\) and \(r(n) \in \{0,1,\dots,n\}\).  

The \((a(n), b(n), r(n))\)-gap Sequentially Invertible Statistical Difference problem, denoted by \(\mathrm{SISD}_{a(n), b(n), r(n)}\), is a promise problem, where the input is a pair of sequences:
\[
\Big(\mathfrak{C}^0 = \{(C^0_{i,\rightarrow}, C^0_{i,\leftarrow})\}_{i \in [\ell]}, \quad \mathfrak{C}^1 = \{(C^1_{i,\rightarrow}, C^1_{i,\leftarrow})\}_{i \in [\ell]} \Big)
\]
where for each \(b \in \{0,1\}\), $\mathfrak{C}^b$ is an \(r(n)\)-invertible circuit sequence on some (identical) number of bits \(k\). Here, \(n := \abs{(\mathfrak{C}^0,\mathfrak{C}^1)}\) is the input size.

For a pair of circuit sequences, we define their output distributions:
\[
D(\mathfrak{C}^b) := \Big\{C^b_{\ell,\rightarrow} \Big(\cdots C^b_{2,\rightarrow} \Big( C^b_{1,\rightarrow} \Big(0^{k(n)}; z_1\Big); z_2\Big); \cdots; z_\ell\Big)\ \Big |\ z_1\xleftarrow{\$}\{0,1\}^{r_1^{b}(n)},\ldots,z_l\xleftarrow{\$}\{0,1\}^{r_l^{b}(n)}\Big\}
\]

The problem is defined as follows:
\begin{itemize}
    \item \(\Pi_\textbf{YES}\) is the set of pairs of sequences such that \(\Delta(D(\mathfrak{C}^0), D(\mathfrak{C}^1)) \leq a(n)\),
    \item \(\Pi_\textbf{NO}\) is the set of pairs of sequences such that \(\Delta(D(\mathfrak{C}^0), D(\mathfrak{C}^1)) > b(n)\).
\end{itemize}
\end{definition}

We define the family of SISD problems where the gap is not too small. 

\begin{definition}
(The Polynomial-Gap Sequentially Invertible Statistical Difference Problem, Definition 4.13 in \cite{bqp_oi}).  
Let \(r : \mathbb{N} \to \mathbb{N} \cup \{0\}\) be a function such that for every \(n \in \mathbb{N}\), \(r(n) \in \{0,1,\dots,n\}\).  The polynomially-bounded sequentially invertible statistical difference problem, denoted \(\mathrm{SISD}_{\texttt{poly},r(n)}\), is a set of promise problems. It is defined as follows:
\[
\mathrm{SISD}_{\texttt{poly},r(n)} := \bigcup_{\substack{a(n), b(n) : \mathbb{N} \to [0,1] \\ \exists \text{ a polynomial } \texttt{poly} : \mathbb{N} \to \mathbb{N} \text{ such that } \forall n \in \mathbb{N}:
b(n)^2 - 2a(n) + a(n)^2 \geq \frac{1}{\texttt{poly}(n)}.}}\{ \mathrm{SISD}_{a(n),b(n),r(n)} \}.
\]
\end{definition}

Finally if the randomness used in each circuit in the sequence is not too large, i.e. $r(n) = \bigo(\log n)$, then this can be solved by a poly time quantum algorithm with access to an oracle computing Order Interference.

\begin{theorem}\label{sisd in bqp^oi}
(\(\mathrm{SISD}_{\texttt{poly},\bigo(\log(n))} \in \mathbf{BQP^{OI}}\), Theorem 5.2 in \cite{bqp_oi}).  
Let \(\mathrm{SISD}_{a,b,r} := (\Pi_\text{YES}, \Pi_\text{NO})\) be a \(\mathrm{SISD}\) promise problem such that \(\mathrm{SISD}_{a,b,r} \in \mathrm{SISD}_{\texttt{poly},O(\log(n))}\). Then, there exists \(\text{M}^{\text{OI}}\), a polynomial time quantum algorithm with oracle access to a computable order interference oracle $\mathcal{OI}$, such that:
\begin{itemize}
    \item For every \(x \in \Pi_\text{YES}\), \(\text{M}^{\text{OI}}(x) = 1\) with probability at least \(p(x)\),
    \item For every \(x \in \Pi_\text{NO}\), \(\text{M}^{\text{OI}}(x) = 0\) with probability at least \(p(x)\),
\end{itemize}
such that \(\forall x \in \Pi_\text{YES}\cup \Pi_\text{NO}, p(x) \geq 1 - 2^{-\texttt{poly}(x)}\) for some polynomial \(\texttt{poly}(\cdot)\).
\end{theorem}

Let us outline the proof of Theorem~\ref{sisd in bqp^oi}. For this it suffices to show that we can produce the output distribution state corresponding to a $\bigo(\log n)$-sequentially invertible circuit family $\mathfrak{C} = \Big\{C_{i,\rightarrow},C_{i,\leftarrow}: \{0,1\}^k\times\{0,1\}^{r_i}\to\{0,1\}^k\Big\}_{i\in[l]}$.

Consider an intermediate state for $i\in[l]$:
\begin{align}
    \ket{C_{1\ldots i}} := \sum_{\substack{z_1\in\{0,1\}^{r_1}\\ \vdots\\z_i\in\{0,1\}^{r_i}}}\ket{C_{i,\rightarrow} \Big(\cdots C_{2,\rightarrow} \Big( C_{1,\rightarrow} \Big(0^{k}; z_1\Big); z_2\Big); \cdots; z_i\Big)} \label{SISD intermediate state}
\end{align}

Clearly, the required output distribution state is the normalization of $\ket{C_{1\ldots l}}$.
Given the (normalization of) state in \eqref{SISD intermediate state}, we will show how to construct $\ket{C_{1\ldots i+1}}$ using the $\mathcal{OI}$ oracle. Note that $C_{i+1,\rightarrow}(\cdot,z_{i+1})$ is efficiently computable and invertible for every (hard-wired) $z_{i+1}\in\{0,1\}^{r_{i+1}}$. Due to this, one can efficiently construct a unitary $U_{i+1,z_{i+1}}$ such that $U_{i+1,z_{i+1}}\ket{x} = \ket{C_{i+1,\rightarrow}(x;z_{i+1})}$ for every $x\in\{0,1\}^k$. From this observation, we have:
\begin{align}
    \ket{C_{1\ldots i+1}} &= \sum_{\substack{z_1\in\{0,1\}^{r_1}\\ \vdots\\z_{i+1}\in\{0,1\}^{r_{i+1}}}}\ket{C_{i+1,\rightarrow} \Big(\cdots C_{2,\rightarrow} \Big( C_{1,\rightarrow} \Big(0^{k}; z_1\Big); z_2\Big); \cdots; z_{i+1}\Big)}\\
    &= \sum_{z_{i+1}\in\{0,1\}^{r_{i+1}}}U_{i+1,z_{i+1}}\ket{C_{1\ldots i}} \label{CI in SISD}
\end{align}

Thus \eqref{CI in SISD} is the normalized choice interference state for the input: $\Big(\{U_{i+1,z_{i+1}}\}_{z_{i+1}\in\{0,1\}^{r_{i+1}}},\ket{C_{1\ldots i}}\Big)$. We will now make use of Lemma~\ref{OI implies CI} to compute this state using the $\mathcal{OI}$ oracle.

To argue that we successfully obtain the required state with high probability, consider the success probability term
$$\Bigg(\frac{\sum_{x\in\{0,1\}^n{\abs{\sum_{z\in[m]}{\alpha_{x,z}}}}}}{\sum_{x\in\{0,1\}^n{\sum_{z\in[m]}{\abs{\alpha_{x,z}}}}}}\Bigg) \Bigg(\frac{\frac{\norm{\text{CI}(\{U_z\}_{z\in[m]},\ket{\psi})}}{m}}{\frac{\norm{\text{CI}(\{U_z\}_{z\in[m]},\ket{\psi})}}{m}+\frac{1}{\lambda}}\Bigg)$$

In this case, $m = 2^{r_{i+1}} = \texttt{poly}(n)$. The other crucial observation is that each unitary $U_{i+1,z_{i+1}}$ is actually a permutation matrix as it maps a classical state $\ket{x}$ to the classical state $\ket{C_{i+1}(x;z_{i+1})}$. Further, it is not hard to see that the state $\ket{\psi} = \ket{C_{1\ldots i}} = \sum_{x\in\{0,1\}^k}\psi_x\ket{x}$ has all amplitudes $\psi_x$ being non-negative. Using these observations, we can show that $\norm{\text{CI}(\{U_z\}_{z\in[m]},\ket{\psi})} \geq \sqrt{m} = \sqrt{2^{r_{i+1}}} = \texttt{poly}(n)$, so by setting $\lambda = \texttt{poly}(n)$ we can get the required state with probability arbitrarily close to $1$, and with time complexity $\texttt{poly(n)}$.

\subsection{Lattice Problems and Variants of Learning With Errors}

A lattice is a discrete additive subgroup of $\mathbb{R}^m$. For an $m$-dimensional lattice $L$, a set of linearly independent vectors $\mathbf{b_1},\ldots,\mathbf{b_n}$ is called a basis of $L$ if $L$ is generated by the set, i.e., $L = \mathbf{B}\mathbb{Z}^n$ where $\mathbf{B} = [\mathbf{b_1},\ldots,\mathbf{b_n}]$. In such a case, the lattice is said to have rank $n$. One can verify that the rank is independent of the choice of basis $B$. 

For an $m$-dimensional lattice $L$ and $\mathbf{x}\in\R^m$ we define $\text{dist}(\mathbf{x},L) = \min\{\norm{\mathbf{x}-\mathbf{v}}:\mathbf{v}\in L\}$

The Gap Closest Vector Problem (GapCVP) is one of the most popular lattice-based promise problems and is believed to be computationally hard (even with quantum computation) for polynomial approximation factor $\gamma(n)$. 

\begin{definition}(GapCVP)
    For an approximation factor $\gamma = \gamma(n)$, an instance of GapCVP$_\gamma$ is given by a lattice $L$, a target vector $\mathbf{t}\in\R^n$ and a number $d > 0$. In YES instances, $\text{dist}(\mathbf{t},L) \leq d$, whereas in NO instances, $\text{dist}(\mathbf{t},L) > \gamma d$.
\end{definition}

The Learning with Errors (LWE) problem was introduced by Oded Regev in 2005 \cite{reg05}. Below we state its search and decisional version.

\begin{definition}(LWE distribution)
    Given a vector $\mathbf{s}\in\Z_q^n$ and a distribution $\chi$ on $\Z$, the LWE distribution $A_{\mathbf{s},\chi}$ consists of samples of the form $(\mathbf{a},b) \in \mathbb{Z}^n_q \times \mathbb{Z}_q$, with $\mathbf{a}\xleftarrow{\$} \mathbb{Z}^n_q$, $b = \langle \mathbf{a}, \mathbf{s} \rangle + e\ (\text{mod }q)$ for $e \leftarrow \chi$. 
\end{definition}

\begin{definition}(Search LWE)
    The input to the search $\text{LWE}^m_{n,q,\chi}$ with dimension $n \geq 1$, modulus $q \geq 2$ and distribution $\chi$ over $\mathbb{Z}$, consists of $m \geq n$ samples from $A_{\mathbf{s},\chi}$ for an unknown $\mathbf{s}\in\Z_q^n$, and the goal is to find $\mathbf{s}$. 
\end{definition}

\begin{definition}(Decisional LWE)
The decisional $\text{LWE}^m_{n,q,\chi}$ with dimension $n \geq 1$, modulus $q \geq 2$ and distribution $\chi$ over $\mathbb{Z}$, asks to distinguish between $m \geq n$ samples from $A_{\mathbf{s},\chi}$ (for a fixed unknown $\mathbf{s}\in\Z_q^n$) and $U(\Z_q^{n+1})$. 
\end{definition}

Let us define the Guassian function and the discrete Gaussian distribution.

\begin{definition}
    The Gaussian function of width $s > 0$ is defined over all $x \in \mathbb{R}^m$ by $\rho_s(x) = e^{-\frac{\|x\|^2}{s^2}}$. For any countable subset $A \subset \mathbb{R}^m$, we denote the Gaussian mass of $A$ by $\rho_s(A) = \sum_{x \in A} \rho_s(x)$.
\end{definition}

\begin{definition}
For a countable set $A \subseteq \R^m$, define the discrete Gaussian distribution on $A$ of width $s$, denoted $D_{A,s}$, such that for $x\in A$, the probability mass at $x$ is $D_{A,s}(x) = \frac{\rho_s(x)}{\rho_s(A)}$. We denote $D_{A,1}$ as simply $D_A$.
\end{definition}

In the most popular setting, the distribution $\chi = D_{\Z,\alpha q}$ is the discrete Gaussian distribution on $\Z$ with width being $\alpha q$ for a suitable parameter $\alpha \in (0,1)$. In this setting, we denote the LWE problem as $\text{LWE}^m_{n,q,\alpha}$.

The following useful fact shows that for any lattice $L$, the mass given by the discrete Gaussian measure $D_{L,r}$ to points of norm greater than $\sqrt{n}r$ is at most exponentially small.

\begin{fact}\label{bound_on_Gaussian_vector}(Lemma 1.5(i) in \cite{lattice_bounds_banaszczyk})
    Let $B_n$ denote the Euclidean unit ball. Then, for any lattice $L$ and any $r > 0$, $\rho_r(L \setminus \sqrt{n}rB_n) < 2^{-2n} \cdot \rho_r(L)$, where $L \setminus \sqrt{n}rB_n$  is the set of lattice points of norm greater than $\sqrt{n}r$.
\end{fact}

A reduction from decisional LWE to search LWE is trivial. For modulus $q$ being the product of distinct (and sufficiently large) $\texttt{poly}(n)$-bounded primes, a reduction from the worst case search LWE to decisional LWE has been shown \cite{peik09}. 

\section{LWE is contained in \textbf{SZK}}

Since its discovery, LWE has served as a security foundation for numerous cryptographic primitives (see eg. an overview in \cite{crypto_primitives_on_LWE}); in particular it is the most prominent hardness assumption for post-quantum cryptography schemes. Although it is folklore in the cryptography community that LWE is contained in $\mathbf{SZK}$, we still devote this section to address this for the sake of completeness \footnote{some citations for this actually point to an SZK protocol for GapCVP (\cite{gapcvp_in_szk}), which is related to LWE}. With this, having a computational ability to efficiently solve $\mathbf{SZK}$ problems will compromise the security of all cryptographic primitives based on the hardness of LWE. 

To establish that decisional LWE is in $\mathbf{SZK}$, we will give a (randomized) Karp reduction from decisional LWE to $\text{GapCVP}_\gamma$ and then use the fact that $\text{GapCVP}_\gamma \in \mathbf{SZK}$ for $\gamma \geq \sqrt{\frac{n}{\bigo(\log n)}}$ \cite{gapcvp_in_szk}. Consider an $\text{LWE}^m_{n,q,\alpha}$ instance $(\mathbf{A},\mathbf{b})$ where $\mathbf{A}\xleftarrow{\$}\Z_q^{m\times n}$, $\mathbf{b}\in\Z_q^{m}$. For a $\text{GapCVP}_\gamma$ instance, we need a lattice $L$, a target vector $\mathbf{t}$, and a distance $d$. We define the following lattice corresponding to $\mathbf{A}$:
$$\Lambda(\mathbf{A}) := \{\mathbf{z}\in\Z^m\ |\ \exists \mathbf{s}\in\Z_q^n, \mathbf{A}\mathbf{s}\equiv\mathbf{z}\ (\text{mod }q)\}$$
and set $L = \Lambda(\mathbf{A})$. The target vector can be set to be $\mathbf{t} = \mathbf{b}$ itself (viewed as an element of $\Z^m$). We set the distance threshold $d = \sqrt{m}\alpha q$.

When $(\mathbf{A},\mathbf{b})$ is from the distribution $\mathcal{A}_{\mathbf{s},D_{\Z,\alpha q}}$ for some $\mathbf{s}\in\Z_q^n$, then $\mathbf{b} = \mathbf{As}+\mathbf{e}$ for $\mathbf{e}\leftarrow D_{\Z,\alpha q}$. In this case, with high probability, $\text{dist}(\Lambda(\mathbf{A}),\mathbf{b}) \leq \norm{\mathbf{e}} \leq \sqrt{m}\alpha q$ where the second inequality is from Fact~\ref{bound_on_Gaussian_vector}.

On the other hand, let us put a lower bound on $\text{dist}(\Lambda(\mathbf{A}),\mathbf{b})$ when $(\mathbf{A},\mathbf{b})\xleftarrow{\$} \Z_q^{m\times (n+1)}$. We want to show that $\text{dist}(\Lambda(\mathbf{A}),\mathbf{b}) > R = \gamma d$ with high probability. To be able to use $\text{GapCVP}_\gamma \in \mathbf{SZK}$, as mentioned above we need $\gamma \geq \sqrt{\frac{n}{\bigo(\log n)}}$, i.e. $R$ needs to be atleast $\sqrt{\frac{n}{c\log n}}\sqrt{m}\alpha q$ for some constant $c>0$. Note that $\mathbf{b}$ is a uniformly random point in $\Z_q^m$, and there are $2^mq^n$ points of the lattice $\Lambda(\mathbf{A})$ which are nearest to the region $\Z_q^m$ ($q^n$ points corresponding to $\mathbf{As}$ for $\mathbf{s}\in\Z_q^n$ in all $2^m$ neighboring cosets of $\Z_q^m$). For any lattice point, the no. of points of $\Z^m$ within distance $R$ can be upper bounded by $(2R)^m$. Thus there are atmost $2^mq^n(2R)^m$ points in $\Z_q^m$ which are at a distance $\leq R$ from $\Lambda(A)$, and $\text{dist}(\Lambda(\mathbf{A}),\mathbf{b}) > R$ except with probability $p = \frac{2^mq^n(2R)^m}{q^m}$. We need this probability to be exponentially small, say $p \leq 2^{-n}$. We can now substitute $R= \sqrt{\frac{n}{c\log n}}\sqrt{m}\alpha q$ to get  
\begin{align*}
    &p = \frac{2^mq^n(2R)^m}{q^m} = \frac{4^mq^n n^{m/2}m^{m/2}\alpha^m}{(c\log n)^{m/2}} \leq \frac{1}{2^n} \\
    \Rightarrow &\alpha \leq \frac{c'\sqrt{\log n}}{2^{n/m}q^{n/m}\sqrt{n}\sqrt{m}}
\end{align*}

For $q = \texttt{poly}(n)$, $m\geq n\log q$, we get that $\alpha = \bigo(\frac{\sqrt{\log n}}{n\sqrt{\log q}})$ suffices. Thus for these parameters, we have established that $\text{LWE}^m_{n,q,\alpha} \in \mathbf{SZK}$.

\begin{remark}
    The above range of parameters for $q,m,\alpha$ overlap with the range where (correct and secure) Public Key Encryption schemes based on hardness of LWE are built (eg. \cite{reg05}). For such parameters, the PKE schemes will get compromised in a computational model where $\mathbf{SZK}$ problems are efficiently decidable. 
    However for some parameter choices, we don't have the $\mathbf{SZK}$ containment. For instance, the specific choice of parameters Regev mentions in \cite{reg05} for his public key cryptosystem, i.e. $q = \Theta(n^2), m = \Theta(n\log n), \alpha = \Theta(\frac{1}{\sqrt{n}\log^2 n})$ do not satisfy the above requirements, in particular $\alpha$ here is larger than required. 
\end{remark}

\begin{remark}
    Bruno et al. \cite{clwe} introduced a continuous version of LWE called CLWE, and its homogeneous variant, hCLWE which is equivalent to CLWE. Bogdanov et. al. built a PKE scheme based on hCLWE \cite{hclwe_in_szk}, and also showed $\text{hCLWE} \in \mathbf{SZK}$ for suitable parameters. Again, their PKE scheme too gets compromised in a computational model where $\mathbf{SZK}$ problems are efficiently decidable. 
\end{remark}

\section{SD is equivalent to SISD}

In this section, we show that the Sequentially Invertible Statistical Difference problem is no easier than the Statistical Difference problem. In fact, we show this for the 1-sequentially invertible case; i.e. we give a poly-time Karp reduction from the statistical difference problem $\text{SD}_{a,b}$ to the 1-sequentially invertible statistical difference problem $\text{SISD}_{a,b,1}$. 

The idea for this reduction is very straightforward and the intuition is as follows: Given circuits $C_0, C_1$ as inputs for SD, we construct 1-sequentially invertible circuit families $\mathfrak{C}^0 = \{C_{i}^0\}_{i\in[l]}, \mathfrak{C}^1 = \{C_{i}^1\}_{i\in[l]}$  in the following manner: The output (and main input) length of these circuits will be $k_I + k$, where $k_I = \max(k_0,k_1)$, i.e. we will aim to capture both an input and an output string in the output of a circuit in the sequence. We will design $\mathfrak{C}^b$ so that a sample from its distribution will be of the form $(x',C_b(x))$ for $x,x'$ being uniformly random strings and independent of each other. Due to this, the probability that $y\in\{0,1\}^k$ occurs as a suffix in a sample is exactly $D(C_b)(y)$ and it is not hard to see that the statistical difference between $C_0$ and $C_1$ will be preserved by $\mathfrak{C}^0$ and $\mathfrak{C}^1$. The only question is how to construct such a sequence using atmost 1 bit of randomness at each step. In the first part of the sequence we will construct samples of the form $(x,C_b(x_{1\ldots k_b}))$ for $x$ being a uniformly random string. This is doable by using 1 bit of randomness at the $i\textsuperscript{th}$ step to form the $i\textsuperscript{th}$ bit of $x$. In the second part, we will utilize the given bit of randomness to `perturb' each bit of $x$ (leaving $C_b(x)$ untouched).

\begin{theorem}\label{sequentially invertible circuits preserving statistical difference}
    Let $C_0:\{0,1\}^{k_0}\to\{0,1\}^k,\ C_1:\{0,1\}^{k_1}\to\{0,1\}^k$ be circuits. Then there exist 1-sequentially invertible circuit families $\mathfrak{C}^0 = \{C^0_{i,\rightarrow},C^0_{i,\leftarrow}\}_{i\in[l]},\ \mathfrak{C}^1 = \{C^1_{i,\rightarrow},C^1_{i,\leftarrow}\}_{i\in[l]}$ such that $\Delta(D(C_0),D(C_1)) = \Delta(D(\mathfrak{C}^0),D(\mathfrak{C}^1))$. In particular $\text{SD}_{a,b} \leq_p \text{SISD}_{a,b,1}$.
\end{theorem}

\begin{proof}
    Let $k_I = \max(k_0,k_1)$. We design circuit families $\mathfrak{C}^0, \mathfrak{C}^1$:
    $$\mathfrak{C}^b = \Big\{C^b_{i,\rightarrow},C^b_{i,\leftarrow}:\{0,1\}^{k_I+k}\times\{0,1\}^{r_i}\times\{0,1\}^{k_I+k}\Big\}_{i\in[l]}$$ for $l = 2k_I+1$, $r_i\leq 1$, $b\in\{0,1\}$.
    
    \begin{itemize}
        \item The first $k_I$ ($i = 1\ \text{to}\ k_I$) circuits $C^b_{i,\rightarrow}:\{0,1\}^{k_I+k}\times\{0,1\}\to\{0,1\}^{k_I+k}$ perform the following: for $x\in\{0,1\}^{k_I}, y\in\{0,1\}^k, z\in\{0,1\}$ 
        $$((x,y),z) \mapsto ((x_1,\ldots,x_{i-1},x_i\oplus z,x_{i+1},\ldots,x_n),y)$$
        \item The $k_I+1$\textsuperscript{th} circuit $C^b_{k_I+1,\rightarrow}:\{0,1\}^{k_I+k}\to\{0,1\}^{k_I+k}$ performs the following: for $x\in\{0,1\}^{k_I}, y\in\{0,1\}^k, z\in\{0,1\}$ 
        $$(x,y) \mapsto (x,y\oplus C_b(x_{1\ldots k_b}))$$
        \item The last $k_I$ circuits $C^b_{i+k_I+1,\rightarrow}:\{0,1\}^{k_I+k}\times\{0,1\}\to\{0,1\}^{k_I+k}$ ($i = 1\ \text{to}\ k_I$) perform the following: for $x\in\{0,1\}^{k_I}, y\in\{0,1\}^k, z\in\{0,1\}$ 
        $$((x,y),z) \mapsto ((x_1,\ldots,x_{i-1},x_i\oplus z,x_{i+1},\ldots,x_n),y)$$
        \item For each $i\in[l]$, $C^b_{i,\leftarrow} = C^b_{i,\rightarrow}$.
    \end{itemize}
    
    Clearly $\mathfrak{C}^0, \mathfrak{C}^1$ are 1-sequentially invertible families. Consider the output distributions corresponding to these families:
    
    $$D(\mathfrak{C}^b) := \Big\{C^b_{\ell,\rightarrow} \Big(\cdots C^b_{2,\rightarrow} \Big( C^b_{1,\rightarrow} \Big(0^{k_I+k}; z_1\Big); z_2\Big); \cdots; z_\ell\Big)\ \Big |\ z_1,\ldots,z_{k_I},z_{k_I+2},\ldots,z_l\xleftarrow{\$}\{0,1\}\Big\}$$
    
    Let $\mathfrak{C}^b_{1\ldots k_I+1}$ denote the subsequence in $\mathfrak{C}^b$ of the first $k_I+1$ circuit pairs. Its distribution is of the form:
    $$D(\mathfrak{C}^b_{1\ldots k_I+1}) = \{(x,C_b(x_{1\ldots,k_b}))\ |\ x\xleftarrow{\$}\{0,1\}^{k_I}\}$$
    
    By the end of the $l$\textsuperscript{th} iteration, this distribution becomes 
    $$D(\mathfrak{C}^b) = \{(x\oplus x',C_b(x_{1\ldots k_b}))\ |\ x\xleftarrow{\$}\{0,1\}^{k_I}, x'\xleftarrow{\$}\{0,1\}^{k_I}\} = \{(x',C_b(x))\ |\ x\xleftarrow{\$}\{0,1\}^{k_b}, x'\xleftarrow{\$}\{0,1\}^{k_I}\}$$
    
    In particular we have 
    $$D(\mathfrak{C}^b)(x,y) := \Prob[(x,y)\leftarrow D(\mathfrak{C}^b)] = \Prob[x\leftarrow\{0,1\}^{k_I}]\cdotp\Prob[y\leftarrow D(C_b)] = \frac{1}{2^{k_I}}\Prob[y\leftarrow D(C_b)]$$
    
    Thus the statistical difference between $D(\mathfrak{C}^0)$ and $D(\mathfrak{C}^1)$ can be calculated as follows:
    \begin{align*}
        \Delta(D(\mathfrak{C}^0),D(\mathfrak{C}^1)) &= \frac{1}{2}\sum_{x\in\{0,1\}^{k_I}}\sum_{y\in\{0,1\}^k}{\abs{D(\mathfrak{C}^0)(x,y)-D(\mathfrak{C}^1)(x,y)}}\\
        &= \frac{1}{2}\sum_{x\in\{0,1\}^{k_I}}\sum_{y\in\{0,1\}^k}{\Bigabs{\frac{1}{2^{k_I}}\Prob[y\leftarrow D(C_0)] - \frac{1}{2^{k_I}}\Prob[y\leftarrow D(C_1)]}}\\
        &= \frac{1}{2}\sum_{y\in\{0,1\}^k}{\Bigabs{\Prob[y\leftarrow D(C_0)] - \Prob[y\leftarrow D(C_1)]}}\ \ = \Delta(D(C_0),D(C_1))
    \end{align*}    
\end{proof}

By combining Theorem~\ref{sequentially invertible circuits preserving statistical difference}, Theorem~\ref{SD_is_SZK_complete} and Theorem~\ref{sisd in bqp^oi} we immediately get the following corollary.

\begin{corollary}($\mathbf{SZK} \subseteq \mathbf{BQP^{OI}}$) For every promise problem $\Pi = (\Pi_{\text{YES}},\Pi_{\text{NO}})\in\mathbf{SZK}$, there exists a polynomial time quantum algorithm \(\text{M}^{\text{OI}}\) with oracle access to a computable order interference oracle $\mathcal{OI}$, such that:
\begin{itemize}
    \item For every \(x \in \Pi_\text{YES}\), \(\text{M}^{\text{OI}}(x) = 1\) with probability at least \(p(x)\),
    \item For every \(x \in \Pi_\text{NO}\), \(\text{M}^{\text{OI}}(x) = 0\) with probability at least \(p(x)\),
\end{itemize}
such that \(\forall x \in \Pi_\text{YES}\cup \Pi_\text{NO}, p(x) \geq 1 - 2^{-\texttt{poly}(x)}\) for some polynomial \(\texttt{poly}(\cdot)\).
\end{corollary}

\section*{Acknowledgements} 
We thank Omri Shmueli for insightful discussions and comments on our draft.

\printbibliography

\end{document}